\documentclass[article, 10pt, conference]{ieeeconf}
\IEEEoverridecommandlockouts                           

\usepackage{graphics} % for pdf, bitmapped graphics files
\usepackage{amsmath} % assumes amsmath package installed
\usepackage{amssymb}  % assumes amsmath package installed
\usepackage{theorem}
\usepackage{color,xspace}
\usepackage{graphicx}
\usepackage{subcaption}
\usepackage[noadjust]{cite}
\usepackage{cleveref}
\usepackage{bm}
\usepackage{bbm}
\usepackage{tikz}
\usepackage{epstopdf} %converting to PDF
\usepackage{algorithm}
\usepackage{algpseudocode}
 \usepackage{booktabs}
 \let\labelindent\relax
\usepackage{enumitem}

\newtheorem{theorem}{Theorem}
\newtheorem{lemma}{Lemma}

\newtheorem{definition}{Definition}

\def\bbx{{\ensuremath{\mathbf x}}}

\def\prob{\mathbf{Pr}}
\def\E{\mathbf{E}}
\def\Cov{\mathbf{Cov}}

\newcommand*{\Scale}[2][4]{\scalebox{#1}{$#2$}}%

\title{\LARGE \bf Probabilistic Verification and Reachability Analysis of \\ Neural Networks via Semidefinite Programming}
\author{Mahyar Fazlyab, Manfred Morari, George J. Pappas
\thanks{$^\dagger$Corresponding author: mahyarfa@seas.upenn.edu. This work was supported by DARPA Assured Autonomy and NSF CPS 1837210. The authors are with the Department of Electrical and Systems Engineering, University of Pennsylvania. Email: \{mahyarfa, morari, pappasg\}@seas.upenn.edu.}}

\begin{document}

\maketitle
\thispagestyle{empty}
\pagestyle{empty}

\begin{abstract}
	Quantifying the robustness of neural networks or verifying
	their safety properties against input uncertainties or
	adversarial attacks have become an important research area
	in learning-enabled systems. Most results concentrate
	around the worst-case scenario where the input of the
	neural network is perturbed within a norm-bounded
	uncertainty set. In this paper, we consider a probabilistic
	setting in which the uncertainty is random with known first
	two moments. In this context, we discuss two relevant
	problems: (i) probabilistic safety verification, in which the goal is to find an upper bound on
	the probability of violating a safety specification; and
	(ii) confidence ellipsoid estimation, in which given a
	confidence ellipsoid for the input of the neural network,
	our goal is to compute a confidence ellipsoid for the
	output. Due to the presence of nonlinear activation
	functions, these two problems are very difficult to solve
	exactly. To simplify the analysis, our main idea is to
	abstract the nonlinear activation functions by a
	combination of affine and quadratic constraints they impose
	on their input-output pairs.  We then show that the safety
	of the abstracted network, which is sufficient for the
	safety of the original network, can be analyzed using
	semidefinite programming. We illustrate the performance of
	our approach with numerical experiments.
	
	%In the proposed SDP, we find the largest confidence ellipsoid in the input space whose reachable set is guaranteed to satisfy the safety specifications. We illustrate the performance of our approach with numerical experiments.
	%Quantifying the robustness of neural networks or verifying their safety properties against input uncertainties or adversarial attacks has become an important area of research in learning-enabled systems. Most results concentrate around the worst-case scenario where the input of the neural network is perturbed within a norm-bounded uncertainty set. In this paper, we consider a probabilistic setting in which the input uncertainty is random with known first two moments and the goal is to estimate the first two moments of the output. We cast this problem as a reachability analysis problem, in which we over-approximate the reachable set of the neural network when the input is within a confidence ellipsoid. More specifically, we find the minimum-volume ellipsoid that encloses the reachable set. 
\end{abstract}

\section{Introduction}
Neural Networks (NN) have been very successful in various applications such as end-to-end learning for self-driving cars \cite{bojarski2016end}, learning-based controllers in robotics \cite{shi2018neural}, speech recognition, and image classifiers. Their vulnerability to input uncertainties and adversarial attacks, however, refutes the deployment of neural networks in safety critical applications. In the context of image classification, for example, it has been shown in several works \cite{zheng2016improving,moosavi2017universal,su2019one} that even adding an imperceptible noise to the input of neural network-based classifiers can completely change their decision.  In this context, verification refers to the process of checking whether the output of a trained NN satisfies certain desirable properties when its input is perturbed within an uncertainty model. More precisely, we would like to verify whether the neural network's prediction remains the same in a neighborhood of a test point $x^\star$. This neighborhood can represent, for example, the set of input examples that can be crafted by an adversary. 

%Uncertainties in the input of neural networks can be either deterministic or stochastic (e.g., additive noise). In the former, we need to study the worst-case performance of the neural network within a bounded uncertainty set. In the latter, we need to adopt a probabilistic approach, in which we need to provide a statistical guarantee on the prediction performance of the neural network subject to random perturbations in their input.

In \emph{worst-case} safety verification, we assume that the input uncertainty is bounded and we verify a safety property for all possible perturbations within the uncertainty set. This approach has been pursued extensively in several works using various tools, such as mixed-integer linear programming \cite{bastani2016measuring,lomuscio2017approach,tjeng2017evaluating}, robust optimization and duality theory \cite{kolter2017provable,dvijotham2018dual}, Satisfiability Modulo Theory (SMT) \cite{pulina2012challenging}, dynamical systems \cite{ivanov2018verisig,xiang2018output}, Robust Control \cite{fazlyab2019safety}, Abstract Interpretation \cite{mirman2018differentiable} and many others \cite{hein2017formal,wang2018efficient}.

In \emph{probabilistic} verification, on the other hand, we assume that the input uncertainty is random but potentially unbounded. Random uncertainties can emerge as a result of, for example, data quantization, input preprocessing, and environmental background noises \cite{weng2018proven}. In contrast to the worst-case approach, there are only few works that have studied verification of neural networks in probabilistic settings \cite{weng2018proven,dvijotham2018verification,bibi2018analytic}.  In situations where we have random uncertainty models, we ask a related question: ``Can we provide statistical guarantees on the output of neural networks when their input  is perturbed with a random noise?" 
In this paper, we provide an affirmative answer by addressing two related problems:
\begin{itemize}[leftmargin=*]
	\item \emph{Probabilistic Verification:} Given a safe region in the output space of the neural network, our goal is estimate the probability that the output of the neural network will be in the safe region when its input is perturbed by a random variable with a known mean and covariance. 
	\item \emph{Confidence propagation:} Given a confidence ellipsoid on the input of the neural network, we want to estimate the output confidence ellipsoid.
\end{itemize}
The rest of the paper is organized as follows. In Section \ref{subsec: output reachable set estimation}, we discuss safety verification of neural networks in both deterministic and probabilistic settings. In Section \ref{sec: Problem Relaxation via Quadratic Constraints}, we provide an abstraction of neural networks using the formalism of quadratic constraints. In Section \ref{sec: Analysis of the Relaxed Network via Semidefinite Programming} we develop a convex relaxation to the problem of confidence ellipsoid estimation. In Section \ref{sec: numerical experiments}, we present the numerical experiments. Finally, we draw our conclusions in Section \ref{sec: conclusions}.

\subsection{Notation and Preliminaries}
We denote the set of real numbers by $\mathbb{R}$, the set of real $n$-dimensional vectors by $\mathbb{R}^n$, the set of $m\times n$-dimensional matrices by $\mathbb{R}^{m\times n}$, and the $n$-dimensional identity matrix by $I_n$. We denote by $\mathbb{S}^{n}$, $\mathbb{S}_{+}^n$, and $\mathbb{S}_{++}^n$ the sets of $n$-by-$n$ symmetric, positive semidefinite, and positive definite matrices, respectively. We denote ellipsoids in $\mathbb{R}^n$ by
$$
\mathcal{E}(x_c,P) = \{x \mid (x-x_c)^\top P^{-1}(x-x_c) \leq 1 \},
%\mathcal{E}(A,b) = \{x \mid \|Ax+b\|_2 \leq 1 \}
$$
where $x_c \in \mathbb{R}^n$ is the center of the ellipsoid and $P \in \mathbb{S}_{++}^n$ determines its orientation and volume. We denote the mean and covariance of a random variable $X \in \mathbb{R}^n$ by $\E[X] \in \mathbb{R}^n$ and $\Cov[X] \in \mathbb{S}_{+}^n$, respectively.
%
%where $P \in \mathbb{S}_{++}^n$ and $x_c \in \mathbb{R}^n$ is the center of the ellipsoid.
%
\section{Safety Verification of Neural Networks} \label{subsec: output reachable set estimation}
\subsection{Deterministic Safety Verification} \label{subsec: Problem Statement}
Consider a multi-layer feed-forward fully-connected neural network described by the following equations,
\begin{align} \label{eq: DNN model 0}
x^0 &=x \\ \nonumber 
x^{k+1} &=\phi(W^k x^k + b^k) \quad k=0, \cdots, \ell-1 \\ \nonumber 
f(x) &= W^\ell x^\ell + b^{\ell},
\end{align}
where $x^0 = x$ is the input to the network, $W^{k} \in \mathbb{R}^{n_{k+1}\times n_k}, \ b^k \in \mathbb{R}^{n_{k+1}}$ are the weight matrix and bias vector of the $k$-th layer. The nonlinear activation function $\phi(\cdot)$ (Rectified Linear Unit (ReLU), sigmoid, tanh, leaky ReLU, etc.) is applied coordinate-wise to the pre-activation vectors, i.e., it is of the form
\begin{align} \label{eq: repeated nonlinearity DNN}
\phi(x) = [\varphi(x_1) \ \cdots \ \varphi(x_{d})]^\top,
\end{align}
where $\varphi$ is the activation function of each individual neuron. Although our framework is applicable to all activation functions, we focus our attention to ReLU activation functions, $\varphi(x)=\max(x,0)$. 

%The output $f(x)$ depends on the specific application we are considering. For example, in image classification with cross-entropy loss, $f(x)$ represents the vector of scores assigned to each image class; or, in feedback control, $x$ is the input to the neural network controller (e.g., tracking error) and $f(x)$ is the control input to the plant.

In deterministic safety verification, we are given a bounded set $\mathcal{X} \subset \mathbb{R}^{n_x}$ of possible inputs (the uncertainty set), which is mapped by the neural network to the output reachable set $f(\mathcal{X})$.
%
%\begin{align} \label{eq: reachable set}
%\mathcal{Y}=f(\mathcal{X}) := \{y \mid y=f(x), \ x \in \mathcal{X}\}.
%\end{align}
%
The desirable properties that we would like to verify can often be described by a set $\mathcal{S} \subset \mathbb{R}^{n_y}$ in the output space of the neural network, which we call the safe region. In this context, the network is safe if $f(\mathcal{X}) \subseteq \mathcal{S}$. 

%A major component of a safety verification algorithm is an efficient method to compute its reachable set. For the case of neural networks, computing the reachable set $\mathcal{Y}$, defined in \eqref{eq: reachable set}, for a given input set $\mathcal{X}$ is difficult due to the presence of nonlinear activation functions. Hence, we need to over-approximate the output reachable set, which we denote by $\tilde{\mathcal{Y}}$. Since $\mathcal{Y}  \subseteq  \tilde{\mathcal{Y}}$, the condition $\tilde{\mathcal{Y}} \subseteq \mathcal{S}$ is sufficient for declaring the network safe.

%
%\begin{example}[Image classification]\normalfont
%	In neural network-based classifiers, $x$ is the image to be classified, and $f(x)$ is the logit input to the classifier (the softmax function) to assign the labels. To assess the robustness of a trained neural network around a test image $x^\star$, suppose we are allowed to change each pixel of the image by $\pm \epsilon$ and we would like to verify whether the set of all perturbed images are classified correctly. In this case, we have $\mathcal{X} = \{x \colon \|x-x^\star\|_{\infty} \leq \epsilon \}$, and $\mathcal{Y}=f(\mathcal{X})$ is the set of all perturbed inputs to the classifier. The safe set $\mathcal{S}$ is the set of all logit values that produce the same label as $x^\star$. Then the condition $\mathcal{Y} \subseteq \mathcal{S}$ guarantees that the network will assign the same label to all images in $\mathcal{X}$ (local robustness).
%\end{example}

\subsection{Probabilistic Safety Verification}
In a deterministic setting, reachability analysis and safety verification is a yes/no problem whose answer does not quantify the proportion of inputs for which the safety is violated. Furthermore, if the uncertainty is random and potentially unbounded, the output $f(x)$ would satisfy the safety constraint only with a certain probability. More precisely, given a safe region $\mathcal{S}$ in the output space of the neural network, we are interested in finding the probability that the neural network maps the random input $X$ to the safe region, %i.e, the quantity
$$
\prob(f(X) \in \mathcal{S}).
$$
Since $f(x)$ is a nonlinear function, computing the distribution of $f(X)$ given the distribution of $X$ is prohibitive, except for special cases. As a result, we settle for providing a \emph{lower bound} $p\in (0,1)$ on the desired probability,
$$
\prob(f(X) \in \mathcal{S})  \geq p.
$$
%Then we have a statistical guarantee that the neural network is safe with probability at least $p$. 
%
%\subsection{Safety Verification of Input Confidence Regions}
%
To compute the lower bound, we adopt a geometrical approach, in which we verify whether the reachable set of a confidence region of the input lies entirely in the safe set $\mathcal{S}$. We first recall the definition of a confidence region.
\begin{definition}[Confidence region]
	The $p$-level ($p \in [0,1]$) confidence region of a vector random variable $X \in \mathbb{R}^n$ is defined as any set $\mathcal{E}_p \subseteq \mathbb{R}^n$ for which
	$
	\prob(X \in \mathcal{E}_p) \geq p.
	$
\end{definition}
Although confidence regions can have different representations, our particular focus in this paper is on ellipsoidal confidence regions. Due to their appealing geometric properties (e.g., invariance to affine subspace transformations), ellipsoids are widely used in robust control to compute reachable sets \cite{wabersich2018linear,van2002conic,cannon2011stochastic}.
%
%This assumption holds for multivariate Gaussian random variables, in particular.

The next two lemmas characterize confidence ellipsoids for Gaussian random variables and random variables with known first two moments.
\begin{lemma}
	Let $X \sim \mathcal{N}(\mu,\Sigma)$ be an $n$-dimensional Gaussian random variable. Then the $p$-level confidence region of $X$ is given by the ellipsoid
	\begin{align} \label{eq: input confidence ellipsoid}
	\mathcal{E}_{p}= \{x \mid (x-\mu)^\top \Sigma^{-1}(x-\mu) \leq \chi^2_{n}(p) \},
	\end{align}
	where $\chi^2_{n}(p)$ is the quantile function of the chi-squared distribution with $n$ degrees of freedom. 
\end{lemma}

For non-Gaussian random variables, we can use Chebyshev's inequality to characterize the confidence ellipsoids, if we know the first two moments.
\begin{lemma} \label{lem: confidence ellipsoid non-Gaussian}
	Let $X$ be an $n$-dimensional random variable with $\E[X]=\mu$ and $\Cov[X]=\Sigma$. Then the ellipsoid
	\begin{align} \label{eq: input confidence ellipsoid 1}
	\mathcal{E}_p= \{x \mid (x-\mu)^\top \Sigma^{-1}(x-\mu) \leq \dfrac{n}{1-p} \},
	\end{align}
	is a $p$-level confidence region of $X$.
\end{lemma}
\begin{lemma} \label{lem: lower bound on safety}
	Let $\mathcal{E}_p$ be a confidence region of a random variable $X$. If $f(\mathcal{E}_p) \subseteq \mathcal{S}$, then $\mathcal{S}$ is a $p$-level confidence region for the random variable $f(X)$, i.e., 
	%\begin{align}
	$\prob(f(X) \in \mathcal{S}) \geq p$.
	%\end{align}
\end{lemma}
\begin{proof}
	The inclusion $f(\mathcal{E}_{p}) \subseteq \mathcal{S}$ implies $\prob(f(X) \in \mathcal{S}) \geq \prob(f(X) \in f(\mathcal{E}_{p}))$. Since $f$ is not necessarily a one-to-one mapping, we have $\prob(f(X) \in f(\mathcal{E}_{p})) \geq \prob(X \in \mathcal{E}_{p})\geq p$. Combining the last two inequalities yields the desired result. 
\end{proof}
\medskip

%Having Lemma \ref{lem: lower bound on safety} at our disposal, we can now study statistical verification problem using a deterministic approach. Explicitly, we need to verify whether the confidence region $\mathcal{E}_{p}$ is mapped to the safe region $\mathcal{S}$ for some $p \in (0,1)$. Then we can assert that, when the neural network input follows the distribution $\mathcal{N}(\mu,\Sigma)$, the output lies within the safe region $\mathcal{S}$ with probability at least $\chi_{n_x}^2(\rho)$. 
According to Lemma \ref{lem: lower bound on safety}, if we can certify that the output reachable set $f(\mathcal{E}_{p})$ lies entirely in the safe set $\mathcal{S}$ for some $p \in (0,1)$, then the network is safe with probability at least $p$. In particular, finding the best lower bound corresponds to the non-convex optimization problem,
\begin{align} \label{eq: maximizing probability}
\mathrm{maximize} \quad p \quad \text{subject to } f(\mathcal{E}_{p}) \subseteq \mathcal{S},
\end{align}
with decision variable $p \in [0,1)$. By Lemma \ref{lem: lower bound on safety}, the optimal solution $p^\star$ then satisfies
\begin{align}
\prob(f(X) \in \mathcal{S})  \geq p^\star.
\end{align} 
\subsection{Confidence Propagation}
A closely related problem to probabilistic safety verification is confidence propagation. Explicitly, given a $p$-level confidence region $\mathcal{E}_p$ of the input of a neural network, our goal is to find a $p$-level confidence region for the output. To see the connection to the probabilistic verification problem, let $\mathcal{S}$ be any outer approximation of the output reachable set, i.e., $f(\mathcal{E}_p) \subseteq \mathcal{S}$. By lemma \ref{lem: lower bound on safety}, $\mathcal{S}$ is a $p$-level confidence region for the output. Of course, there is an infinite number of such possible confidence regions. Our goal is find the ``best" confidence region with respect to some metric. Using the volume of the ellipsoid as an optimization criterion, the best confidence region amounts to solving the problem
\begin{align} \label{eq: min vol E}
\mathrm{minimize} \ \mathrm{Volume}(\mathcal{S}) \quad \text{subject to } f(\mathcal{E}_p) \subseteq \mathcal{S}. 
\end{align}
The solution to the above problem provides the $p$-level confidence region with the minimum volume. Figure \ref{fig:confidence estimation} illustrates the procedure of confidence estimation. In the next section, we provide a convex relaxation of the optimization problem \eqref{eq: min vol E}. The other problem in \eqref{eq: maximizing probability} is a straightforward extension of confidence estimation, and hence, we will not discuss the details.
%In this paper we consider the confidence propagation problem \eqref{eq: min vol E}. The other problem is a straightforward extension of confidence estimation, and hence, we will not discuss the details.

%Having characterized the confidence region $\mathcal{E}_{p}$ of $X$, our main idea is to consider $\mathcal{E}_{p}$ as the set of all possible inputs to the neural network. If we can certify that the image $f(\mathcal{E}_{p})$ of the ellipsoid under $f$ lies entirely in the safe set $\mathcal{S}$, then we have a statistical guarantee that the network is safe with probability at least $p$. More formally, we can state the following result.
%

%
%where $\rho \geq 0$. It follows by definition that a random variable $X \sim \mathcal{N}(\mu,\Sigma)$ lies in the ellipsoid $\mathcal{E}_{p}$ with probability
%$$\prob(X \in \mathcal{E}_{p}) = \chi_{n_x}^2(\rho),$$
%where  $\rho \mapsto \chi_{n_x}^2(\rho)$ is the cumulative distribution function (CDF) of the chi-squared distribution with $n_x$ degrees of freedom. Note that $\chi_{n_x}^2(\rho)$ is a monotonically increasing function of $\rho$, $\chi_{n_x}^2(0)=0$,  and $\chi_{n_x}^2(\infty)=1$.
%
%
%
%

\begin{figure}
	\centering
	\includegraphics[width=\linewidth]{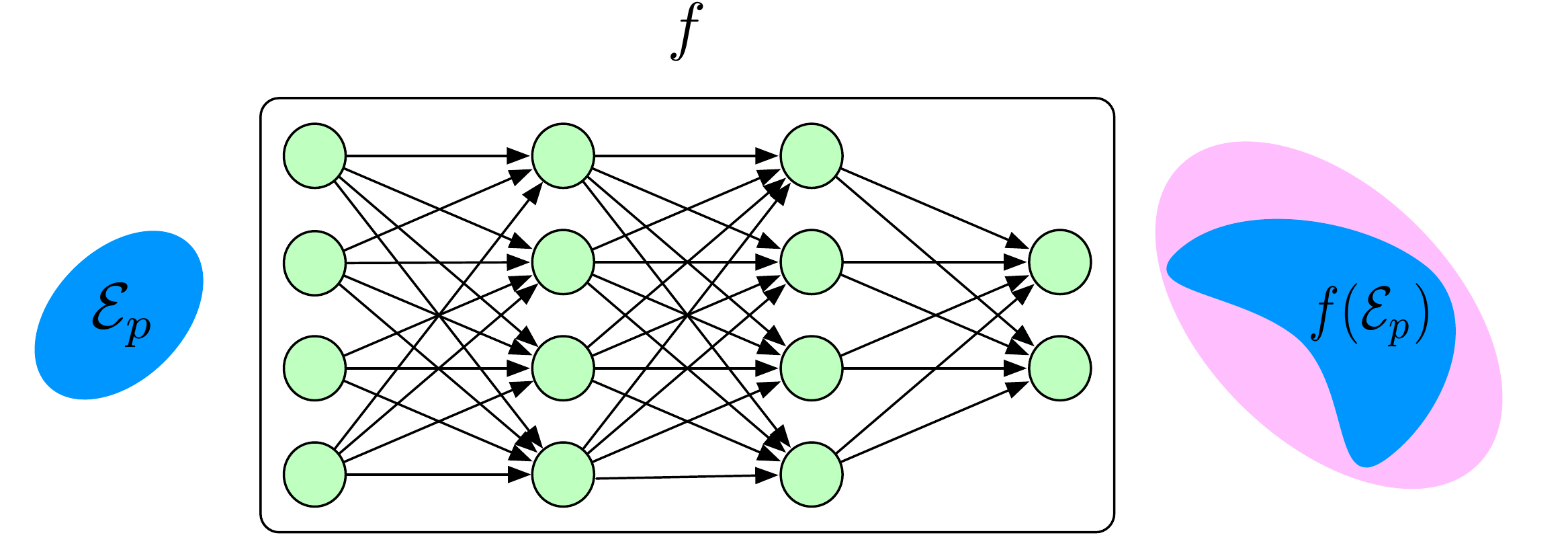}
	\caption{\small $p$-level input confidence ellipsoid $\mathcal{E}_p$, its image $f(\mathcal{E}_p)$, and the estimated output confidence ellipsoid.}
	\label{fig:confidence estimation}
\end{figure}

\section{Problem Relaxation via Quadratic Constraints} \label{sec: Problem Relaxation via Quadratic Constraints}
Due to the presence of nonlinear activation functions, checking the condition $f(\mathcal{E}_{p}) \subseteq \mathcal{S}$ in \eqref{eq: maximizing probability} or \eqref{eq: min vol E} is a non-convex feasibility problem and is NP-hard, in general. Our main idea is to abstract the original network $f$ by another network $\tilde{f}$ in the sense that $\tilde{f}$ over-approximates the output of the original network for any input ellipsoid, i.e., $f(\mathcal{E}_{p}) \subseteq \tilde f(\mathcal{E}_{p})$ for any $p \in [0,1)$. Then it will be sufficient to verify the safety properties of the relaxed network, i.e., verify the inclusion $\tilde f(\mathcal{E}_{p}) \subseteq \mathcal{S}$. In the following, we use the framework of quadratic constraints to develop such an abstraction.

%%In the next section, we develop a convex feasibility problem, in terms of a linear matrix inequality (LMI), that verifies whether $f(\mathcal{E}_{p}) \subseteq \mathcal{S}$ for a given $p \in [0,1)$.

%\subsection{Set Representation by Quadratic Forms}
%We begin by noting that we can describe the input ellipsoid $\mathcal{E}_{p}$ and the safe set $\mathcal{S}$ as given in \eqref{eq: safe set half space} by quadratic constraints.
%%
%%\begin{align}
%%\mathcal{E}_{p} = \left\{x \mid \begin{bmatrix}
%%x \\ 1
%%\end{bmatrix}^\top \begin{bmatrix}
%%\Sigma^{-1}  & -\Sigma^{-1} \mu \\ -\mu^\top \Sigma^{-1} & \mu^\top \Sigma^{-1} \mu - \rho
%%\end{bmatrix} \begin{bmatrix}
%%x \\ 1
%%\end{bmatrix} \leq 0 \right\},
%%\end{align}
%%
%\begin{align}
%\mathcal{E}_{p} = \left\{x \mid \begin{bmatrix}
%x \\ 1
%\end{bmatrix}^\top E(p) \begin{bmatrix}
%x \\ 1
%\end{bmatrix} \leq 0 \right\},
%\end{align}
%where $E(p)$ is a symmetric matrix indexed by $p \in (0,1)$.
%%
%%Formally, $\mathcal{E}_{p}$ can be described by
%%%
%%%
%%%where the symmetric matrix $E(\rho)$ is given by
%%%$$
%%%E(\rho) =  \begin{bmatrix}
%%%\Sigma^{-1}  & -\Sigma^{-1} \mu \\ -\mu^\top \Sigma^{-1} & \mu^\top \Sigma^{-1} \mu - \rho
%%%\end{bmatrix}
%%%$$
%Similarly, we assume that the safe set $\mathcal{S}$ can be described as
%%
%\begin{align} \label{eq: safety set spec}
%\mathcal{S} =  \left\{y \mid \begin{bmatrix}
%y \\ 1
%\end{bmatrix}^\top S_i \begin{bmatrix}
%y \\ 1
%\end{bmatrix} \leq 0 \quad i=1,\cdots,m \right\}.
%\end{align}
%%
%for some given symmetric matrix $S$.

\subsection{Relaxation of Nonlinearities by Quadratic Constraints}
In this subsection, we show how we can abstract activation functions, and in particular the ReLU function, using quadratic constraints. We first provide a formal definition, introduced in \cite{fazlyab2019safety}.
\begin{definition} \label{def: QC}
	Let $\phi \colon \mathbb{R}^{d} \to \mathbb{R}^{d}$ be and suppose $\mathcal{Q} \subset \mathbb{S}^{2d+1}$ is the set of all symmetric and indefinite matrices $Q$ such that the inequality
	\begin{align} \label{eq: QC def}
	\begin{bmatrix}
	x \\ \phi(x) \\ 1
	\end{bmatrix}^\top  Q  \begin{bmatrix}
	x \\ \phi(x) \\ 1
	\end{bmatrix} \geq 0,
	\end{align}
	holds for all $x \in \mathbb{R}^d$. Then we say $\phi$ satisfies the quadratic constraint defined by $\mathcal{Q}$.
\end{definition}
Note that the matrix $Q$ in Definition \ref{def: QC} is indefinite, or otherwise, the constraint trivially holds. Before deriving QCs for the ReLU function, we recall some definitions, which can be found in many references; for example \cite{megretski1997system,d2001new}.
\begin{definition}[Sector-bounded nonlinearity]
	A nonlinear function $\varphi \colon \mathbb{R} \to \mathbb{R}$ is \emph{sector-bounded} on the sector $[\alpha,\beta]$ ($0 \leq \alpha \leq \beta$) if the following condition holds for all $x$,
	\begin{align} \label{eq: sector bound}
	(\varphi(x)-\alpha x)(\varphi(x)-\beta x) \leq 0.
	\end{align}
\end{definition}
\begin{definition}[Slope-restricted nonlinearity]
	A nonlinear function $\varphi(x) \colon \mathbb{R} \to \mathbb{R}$ is slope-restricted on $[\alpha,\beta]$ ($0 \leq \alpha \leq \beta$) if 	for any  $(x,\varphi(x))$ and $(x^\star,\varphi(x^\star))$,%, where $\varphi = \varphi(x)$ and $\varphi^\star=\varphi(x^\star)$ with an abuse of notation.
	\begin{align} \label{eq: quadratic constraint 0}
	(\varphi(x)\!-\!\varphi(x^{\star})-\alpha (x\!-\!x^{\star}))(\varphi(x)\!-\!\varphi(x^{\star})-\beta (x\!-\!x^{\star})) \leq 0.
	\end{align}

\end{definition}
%

%\medskip 
%Note that \eqref{eq: quadratic constraint 0} is equivalent to
%%
%\begin{align} \label{eq: quadratic constraint 3}
%\alpha \leq \dfrac{\varphi(x)-\varphi(x^\star)}{x-x^\star} \leq \beta \quad \forall x,x^\star \in \mathbb{R}.
%\end{align}
%%
%which states that the chord connecting any two points on the curve of the function $y=\varphi(x)$ has a slope that is at least $\alpha$ and at most $\beta$. If the slope restriction condition holds only when the point $(x^\star,\varphi(x^\star))$ is fixed, then the nonlinearity is sector-bounded in the sector $[0,1]$. 
%\medskip
%

\textit{Repeated nonlinearities.} Assuming that the same activation function is used in all neurons, we can exploit this structure to refine the QC abstraction of the nonlinearity. Explicitly, suppose $\varphi \colon \mathbb{R} \to \mathbb{R}$ is slope-restricted on $[\alpha,\beta]$ and let $\phi(x)=[\varphi(x_1)$ $\cdots \varphi(x_d)]^\top$ be a vector-valued function constructed by component-wise repetition of $\varphi$. 
%
%Given two pairs $(x,\phi)$ and $(x^\star,\phi^\star)$, we can write the following constraint for $i=1,\cdots,d$:
%
%$$\begin{bmatrix}
%x_1 \\ \cdots \\ x_d \\ 1
%\end{bmatrix} \overset{\phi}{\mapsto} \begin{bmatrix}
%\varphi(x_1) \\ \cdots \\ \varphi(x_d) \\ 1
%\end{bmatrix}$$ 
%
%be a concatenation of the component-wise repetition of $\varphi$ and a constant bias term. We can therefore write the following constraint for each coordinate of the arbitrary input-output pairs $(x,\phi(x)) ,\ (x^\star,\phi(x^\star))$:
%
%\begin{align} \label{eq: repeated 1}
%(\phi_i \!- \! \phi_i^{\star} \!- \! \alpha (x_i-x_i^{\star}))(\phi_i-\phi_i^{\star}\!-\!\beta (x_i-x_i^{\star})) \leq 0.
%\end{align}
%
%By adding the above inequalities for all $i$, we obtain the quadratic constraint in \eqref{eq: quadratic constraint 0}, implying that $\phi$ (like $\varphi$), is slope-restricted on $[\alpha,\beta]$. 
It is not hard to verify that $\phi$ is also slope-restricted in the same sector. However, this representation simply ignores the fact that all the nonlinearities that compose $\phi$ are the same. By taking advantage of this structure, we can refine the quadratic constraint that describes $\phi$. To be specific, for an input-output pair $(x,\phi(x)), \ x \in \mathbb{R}^d$, we can write the slope-restriction condition
\begin{align} \label{eq: repeated 2}
\Scale[0.99]{(\varphi(x_i) \!- \! \varphi(x_j) \!-\!\alpha (x_i\!-\!x_j))(\varphi(x_i) \!- \! \varphi(x_j) \!- \! \beta (x_i \!- \! x_j)) \! \leq \! 0,}
\end{align}
for all distinct $i,j$. This particular QC can tighten the relaxation incurred by the QC abstraction of the nonlinearity.

%can considerably reduces conservatism, especially for deep networks, as it reasons about \emph{the coupling between the neurons throughout the entire network}. By making an analogy to dynamical systems, we can interpret the neural network as a time-varying discrete-time dynamical system where the same nonlinearity is repeated for all time indexes $k$ (the layer number). Then the QC in \eqref{eq: repeated 2} couples all the possible neurons. 
%
There are several results in the literature about repeated nonlinearities. For instance, in \cite{d2001new,kulkarni2002all}, the authors derive QCs for repeated and odd nonlinearities (e.g. tanh function).
\subsection{QC for ReLU function}
In this subsection, we derive quadratic constraints for the ReLU function, $\phi(x) = \max(0,x), \ x \in \mathbb{R}^d$. Note that this function lies on the boundary of the sector $[0,1]$. More precisely, we can describe the ReLU function by three quadratic and/or affine constraints:
\begin{align} \label{eq: relu qcs}
y_i = \max(0,x_i) \Leftrightarrow y_i \geq x_i, \ \ y_i \geq 0, \ \ y_i^2 = x_iy_i.
\end{align}
On the other hand, for any two distinct indices $i \neq j$, we can write the constraint \eqref{eq: repeated 2} with $\alpha=0$, and $\beta=1$,
\begin{align} \label{eq: relu qcs 1}
(y_j-y_i)^2 \leq (y_j-y_i)(x_j-x_i).
\end{align}
By adding a weighted combination of all these constraints (positive weights for inequalities), we find that the ReLU function $y = \max(0,x)$ satisfies
\begin{align} \label{eq: weighted combination}
&\sum_{i=1}^{d} \lambda_i (y_i^2-x_i y_i) + \nu_i (y_i-x_i) + \eta_i y_i - \\
&\qquad \quad \sum_{i \neq j} \lambda_{ij} \left((y_j-y_i)^2 - (y_j-y_i)(x_j-x_i) \right) \geq 0, \nonumber 
\end{align}
for any multipliers $(\lambda_i,\nu_i,\eta_i,\lambda_{ij}) \in \mathbb{R} \times \mathbb{R}_{+}^3$ for $i,j \in \{1,\cdots,d\}$. This inequality can be written in the compact form \eqref{eq: QC def}, as stated in the following lemma.
\begin{lemma}[QC for ReLU function]  \label{lem: QC for relu}
	The ReLU function, $\phi(x) = \max(0,x) \colon \mathbb{R}^d \to \mathbb{R}^d$, satisfies the QC defined by $\mathcal{Q}$ where
	\begin{align}  \label{lem: QC for relu 1}
	\mathcal{Q}=\left\{ Q\mid Q= \begin{bmatrix}
	0 & T & -\nu \\ T & -2T & \nu \!+ \! \eta\\ -\nu^\top & \nu^\top \!+ \! \eta^\top & 0
	\end{bmatrix}\right\}.
	\end{align}
	Here $\eta,\nu \geq 0$ and $T \in \mathbb{S}_{+}^d$ is given by
	$$T = \sum_{i=1}^d \lambda_i e_i e_i^\top + \sum_{i=1}^{d-1} \sum_{j>i}^{d} \lambda_{ij}(e_i-e_j)(e_i-e_j)^\top,$$
	where $e_i$ is the $i$-th basis vector in $\mathbb{R}^d$ and $\lambda_{ij} \geq 0$.
\end{lemma}
\begin{proof} See \cite{fazlyab2019safety}.
\end{proof}
\vspace{1em}
Lemma \ref{lem: QC for relu} characterizes a family of valid QCs for the ReLU function. It is not hard to verify that the set $\mathcal{Q}$ of valid QCs is a convex cone. As we will see in the next section, the matrix $Q$ in \eqref{lem: QC for relu 1} appears as a decision variable in the optimization problem.

\subsection{Tightening the Relaxation}
In the previous subsection, we derived QCs that are valid for the whole space $\mathbb{R}^d$. When restricted to a region $\mathcal{R} \subseteq \mathbb{R}^d$, we can tighten the QC relaxation. Consider the relationship $\phi(x) = \max(0,x), \ x \in \mathcal{R} \subseteq \mathbb{R}^d$ and
%
%In this subsection, we derive quadratic constraints for the ReLU function, $\phi(x) = \max(0,x) \ x \in \mathbb{R}^d$ over some region $\mathcal{R} \subseteq \mathbb{R}^d$. 
let $\mathcal{I}^{+}$, and $\mathcal{I}^{-}$ be the set of neurons that are always active or always inactive, i.e.,
\begin{align}
\mathcal{I}^{+} &= \{i  \mid x_i \geq 0 \text{ for all } x \in \mathcal{R}\} \\ \notag
\mathcal{I}^{-} &= \{i  \mid x_i < 0 \text{ for all } x \in \mathcal{R}\}. %\\ \notag
%\mathcal{I}^{\pm}&= \mathcal{I} \setminus (\mathcal{I}^{+} \cup \mathcal{I}^{-}).
\end{align}
The constraint $y_i \geq x_i$ holds with equality for active neurons. %Therefore, the corresponding multiplier $\nu_i$ in \eqref{eq: weighted combination} can be arbitrary and not necessarily nonnegative. 
Therefore, we can write
\begin{align*}
\nu_i \in \mathbb{R} \text{ if } i \in \mathcal{I}^{+}, \ \nu_i \geq 0 \text{ otherwise}.
\end{align*}
Similarly, the constraint $y_i \geq 0$ holds with equality for inactive neurons. Therefore, we can write
\begin{align*}
\eta_i \in \mathbb{R} \text{ if } i \in \mathcal{I}^{-}, \ \eta_i \geq 0 \text{ otherwise}.
\end{align*}
Finally, it can be verified that the cross-coupling constraint in \eqref{eq: relu qcs 1} holds with equality for pairs of always active or always inactive neurons. Therefore, for any $1 \leq i < j \leq d$, we can write
\begin{alignat*}{2}
\lambda_{ij} &\in \mathbb{R} &&\text{ if } (i,j) \in \mathcal{I}^{+}\times \mathcal{I}^{+} \text{ or } (i,j) \in \mathcal{I}^{-}\times \mathcal{I}^{-} \\
%\lambda_{ij} &=0 &&\text{ if } (i,j) \in \mathcal{I}^{+}\times \mathcal{I}^{-} \text{ or } (i,j) \in \mathcal{I}^{-}\times \mathcal{I}^{+} \\
\lambda_{ij} &\geq 0 &&\text{ otherwise}.
\end{alignat*}
These additional degrees of freedom on the multipliers can tighten the relaxation incurred in \eqref{eq: weighted combination}. Note that the set of active or inactive neurons are not known \textit{a priori}. However, we can partially find them using, for example, interval arithmetic. 
%To be more specific, suppose $\ell \leq x \leq u$. Then, it can be easily shown that
%%
%\begin{align}
%\min(W,0) \ell + b\leq Wx+b \leq \max(W,0) u + b
%\end{align}
%%
%Since $\phi$ is monotone, we can also write
%%
%\begin{align}
%\phi(\min(W,0) \ell + b) \leq \phi(Wx+b) \leq \phi(\max(W,0) u + b)
%\end{align}
%%
%In other words, we can bound the output of the map $x \mapsto \phi(Wx+b)$ using bounds on its input. Using these bounds, we can partially find the sets $\mathcal{I}^{+}$ and $\mathcal{I}^{-}$.
%%
%\begin{align} \label{eq: relu qcs}
%y_i = \max(0,x_i) \Leftrightarrow y_i \geq x_i, \ \ y_i \geq 0, \ \ y_i^2 = x_iy_i.
%\end{align}
%
%
%
%
\section{Analysis of the Relaxed Network via Semidefinite Programming} \label{sec: Analysis of the Relaxed Network via Semidefinite Programming}
In this section, we use the QC abstraction developed in the previous section to analyze the safety of the relaxed network. In the next theorem, we state our main result for one-layer neural networks and will discuss the multi-layer case in Section \ref{sec: Multi-layer Case}.

\ifx
In the next subsection, we address probabilistic verification when the safe set is described by the . We first consider the case where the polytope is the half-space $\mathcal{S}=\{y \mid a^\top y -b \leq 0 \}$. The more general case of polytopes \eqref{eq: safe set half space} is a straightforward extension, which we wil discuss after stating the main result.
\begin{theorem}
	Consider a one-layer neural network $f \colon \mathbb{R}^{n_x} \to \mathbb{R}^{n_y}$ described by the equations
	\begin{align}
	f(x)= W^1 \phi(W^0 x + b^0) + b^1,
	\end{align}
	where $\phi \colon \mathbb{R}^{n_1} \to \mathbb{R}^{n_1}$ is a nonlinear activation function that satisfies the quadratic constraint defined by $\mathcal{Q}$, i.e., for any $Q \in \mathcal{Q}$,
	\begin{align} \label{thm: main result one layer 1}
	\begin{bmatrix}
	z \\ \phi(z) \\ 1
	\end{bmatrix}^\top Q     \begin{bmatrix}
	z \\ \phi(z) \\ 1
	\end{bmatrix} \geq 0 \quad \text{for all } z.
	\end{align}
	Suppose the input to the neural network is a random variable with mean and covariance $\mu_x \in \mathbb{R}^n_x, \Sigma_x \in \mathbb{S}_{++}^{n_x}$. Consider the following matrix inequality:
	\begin{align}  \label{thm: main result polytope}
	M_1(\tau) + M_2(Q) + M_3(S) \preceq 0.
	\end{align}
	where 
	\begin{subequations}
		\begin{align*}  %\label{thm: main result one layer 3}
		M_1(\tau) &\!=\!   \begin{bmatrix}
		I_{n_x} & 0 \\ 0 & 0 \\ 0 & 1
		\end{bmatrix} P(\tau) \begin{bmatrix}
		I_{n_x} & 0 \\ 0 & 0 \\ 0 & 1
		\end{bmatrix}^\top \\
		M_2(Q) &= \begin{bmatrix}
		{W^0}^\top  & 0 & 0 \\ 0 & I_{n_1} & 0 \\ {b^0}^\top & 0 & 1
		\end{bmatrix} Q  \begin{bmatrix}
		{W^0}  & 0 & b^0 \\ 0 & I_{n_1} & 0 \\ 0 & 0 & 1
		\end{bmatrix} \\
		M_3(S) &=\begin{bmatrix}
		0 & 0 \\ {W^1}^\top & 0 \\ {b^1}^\top & 1
		\end{bmatrix} S \begin{bmatrix}
		0 & W^1 & b^1 \\ 0 & 0 & 1
		\end{bmatrix}
		\end{align*}
	\end{subequations}
		and 
		\begin{align*}
		P(\tau) &= \tau \begin{bmatrix}
		-\Sigma_x^{-1}  & \Sigma_x^{-1} \mu_x \\ \mu_x^\top \Sigma_x^{-1} & -\mu_x^\top \Sigma_x^{-1} \mu_x+\rho
		\end{bmatrix} \ S = \begin{bmatrix}
		0 & a \\ a^\top &  -2b
		\end{bmatrix} 
		\end{align*}
	If \eqref{thm: main result polytope} is feasible for some $(\tau,Q) \in \mathbb{R}_{+} \times \mathcal{Q}$, then
	$$
	%f(\mathcal{E}) \subseteq \mathcal{S}
	\prob(a^\top f(X)\leq b) \geq 1 - \frac{n_y}{\rho}.
	$$
\end{theorem}
\fi

\begin{theorem}[Output covering ellipsoid] \label{thm: main result one layer}
	Consider a one-layer neural network $f \colon \mathbb{R}^{n_x} \to \mathbb{R}^{n_y}$ described by the equation
	\begin{align}
	y= W^1 \phi(W^0 x + b^0) + b^1,
	\end{align}
	where $\phi \colon \mathbb{R}^{n_1} \to \mathbb{R}^{n_1}$ satisfies the quadratic constraint defined by $\mathcal{Q}$, i.e., for any $Q \in \mathcal{Q}$,
	\begin{align} \label{thm: main result one layer 1}
	\begin{bmatrix}
	z \\ \phi(z) \\ 1
	\end{bmatrix}^\top Q     \begin{bmatrix}
	z \\ \phi(z) \\ 1
	\end{bmatrix} \geq 0 \quad \text{for all } z.
	\end{align}
	Suppose $x \in \mathcal{E}(\mu_x,\Sigma_x)$. Consider the following matrix inequality
	\begin{align}  \label{thm: main result one layer 2}
	M_1 + M_2+ M_3 \preceq 0,
	\end{align}
	where 
	\begin{subequations}
		\begin{align*}  %\label{thm: main result one layer 3}
		M_1 &\!=\!   \begin{bmatrix}
		I_{n_x} & 0 \\ 0 & 0 \\ 0 & 1
		\end{bmatrix} P(\tau) \begin{bmatrix}
		I_{n_x} & 0 \\ 0 & 0 \\ 0 & 1
		\end{bmatrix}^\top \\
		M_2 &= \begin{bmatrix}
		{W^0}^\top  & 0 & 0 \\ 0 & I_{n_1} & 0 \\ {b^0}^\top & 0 & 1
		\end{bmatrix} Q  \begin{bmatrix}
		{W^0}  & 0 & b^0 \\ 0 & I_{n_1} & 0 \\ 0 & 0 & 1
		\end{bmatrix} \\
		M_3 &=\begin{bmatrix}
		0 & 0 \\ {W^1}^\top & 0 \\ {b^1}^\top & 1
		\end{bmatrix} S(A,b) \begin{bmatrix}
		0 & W^1 & b^1 \\ 0 & 0 & 1
		\end{bmatrix}
		\end{align*}
	\end{subequations}
	with
	\begin{align*}
	P(\tau) &= \tau \begin{bmatrix}
	-\Sigma_x^{-1}  & \Sigma_x^{-1} \mu_x \\ \mu_x^\top \Sigma_x^{-1} & -\mu_x^\top \Sigma_x^{-1} \mu_x+1
	\end{bmatrix} \\\
	S(A,b) &= \begin{bmatrix}
	A^2 & A b \\ b^\top A &  b^\top b - 1
	\end{bmatrix}.
	\end{align*}
	
	If \eqref{thm: main result one layer 2} is feasible for some $(\tau,A,Q,b) \in \mathbb{R}_{+} \times \mathbb{S}^{n_y} \times \mathcal{Q} \times \mathbb{R}^{n_y}$, then $y \in \mathcal{E}(\mu_y,\Sigma_y)$ with
	$
	\mu_y = -A^{-1}b  \emph{ and } \Sigma_y = A^{-2}.
	$ %$\mathcal{Y} \cap \neg \mathcal{S} = \emptyset$.
\end{theorem}
\begin{proof}
	We first introduce the auxiliary variable $z$, and rewrite the equation of the neural network as
	\begin{align*}
	z  &= \phi(W^0 x + b^0) \quad y = W^1 z + b^1.
	\end{align*}
	Since $\phi$ satisfies the QC defined by $\mathcal{Q}$, we can write the following QC from the identity $z = \phi(W^0 x + b^0)$:
	\begin{align} \label{thm: hyperplance one layer 7}
	\begin{bmatrix}
	W^0 x + b^0 \\ z \\ 1
	\end{bmatrix}^\top Q \begin{bmatrix}
	W^0 x +b^0 \\ z \\ 1
	\end{bmatrix} \geq 0, \ \text{for all } Q \in \mathcal{Q}.
	\end{align}
	By substituting the identity
	\begin{align*} %\label{thm: hyperplance one layer 8}
	\begin{bmatrix}
	W^0 x + b^0 \\ z \\ 1
	\end{bmatrix} = \begin{bmatrix}
	W^0 & 0 & b^0 \\ 0 & I_{n_1} & 0 \\ 0 & 0 & 1
	\end{bmatrix} \begin{bmatrix}
	x \\ z \\ 1
	\end{bmatrix},
	\end{align*}
	back into \eqref{thm: hyperplance one layer 7} and denoting $\bbx = [{x}^\top \ {z}^\top]^\top$, we can write the inequality
	\begin{align} \label{thm: hyperplance one layer 8.5}
	\begin{bmatrix}
	\bbx \\ 1
	\end{bmatrix}^\top 
	M_2\begin{bmatrix}
	\bbx \\ 1
	\end{bmatrix} \geq 0,
	%
	% \bbx \geq 0, \quad \text{for all } Q \in \mathcal{Q}.
	\end{align}
	for any $Q \in \mathcal{Q}$ and all $\bbx$. By definition, for all $x \in \mathcal{E}(\mu_x,\Sigma_x)$, we have
	%\begin{align*}
	$(x-\mu_x)^\top \Sigma_x^{-1} (x-\mu_x) \leq 1$, 
	%\end{align*}
	which is equivalent to
	\begin{align*} %\label{thm: hyperplance one layer 9}
	\tau \begin{bmatrix}
	x  \\ 1
	\end{bmatrix}^\top  \begin{bmatrix}
	-\Sigma_x^{-1}  & \Sigma_x^{-1} \mu_x \\ \mu_x^\top \Sigma_x^{-1} & -\mu_x^\top \Sigma_x^{-1} \mu_x +1 
	\end{bmatrix} \begin{bmatrix}
	x  \\ 1
	\end{bmatrix} \geq 0.
	\end{align*}
   By using the identity
   \begin{align*}
   \begin{bmatrix}
   x  \\ 1
   \end{bmatrix} = \begin{bmatrix}
   I_{n_x} & 0 & 0 \\ 0 & 0 & 1
   \end{bmatrix} \begin{bmatrix}
   x \\ z \\ 1
   \end{bmatrix},
   \end{align*}
   we conclude that for all $x \in \mathcal{E}(\mu_x,\Sigma_x), \ z = \phi(W^0x+b)$,
	\begin{align} \label{thm: hyperplance one layer 10}
	\begin{bmatrix}
	\bbx \\ 1
	\end{bmatrix}^\top 
	M_1\begin{bmatrix}
	\bbx \\ 1
	\end{bmatrix}  \geq 0. %\quad \text{for all } x \in \mathcal{E}(\mu_x,\Sigma_x).
	\end{align}
	Suppose \eqref{thm: main result one layer 2} holds for some $(A,Q,b) \in \mathbb{S}^{n_y} \times \mathcal{Q} \times \mathbb{R}^{n_y}$. By left- and right- multiplying both sides of \eqref{thm: main result one layer 1} by $[\bbx^\top \ 1]$ and $[\bbx^\top \ 1]^\top$, respectively, we obtain
	\begin{align*}
	\begin{bmatrix}
		\bbx \\ 1
		\end{bmatrix}^\top M_{1}    \begin{bmatrix}
		\bbx \\ 1
		\end{bmatrix} + \begin{bmatrix}
		\bbx \\ 1 \end{bmatrix}^\top M_{2} \begin{bmatrix}
		\bbx \\ 1 \end{bmatrix} + \begin{bmatrix}
	\bbx \\ 1 \end{bmatrix}^\top  M_{3} \begin{bmatrix}
	\bbx \\ 1 \end{bmatrix} \leq 0.
	\end{align*}
	For any $x \in \mathcal{E}(\mu_x,\Sigma_x)$ the first two quadratic terms are nonnegative by \eqref{thm: hyperplance one layer 10} and \eqref{thm: hyperplance one layer 8.5}, respectively. Therefore, the last term on the left-hand side must be nonpositive for all $x \in \mathcal{E}(\mu_x,\Sigma_x)$,
	\begin{align*}
	\begin{bmatrix}
	\bbx \\ 1
	\end{bmatrix}^\top 
	M_3\begin{bmatrix}
	\bbx \\ 1
	\end{bmatrix}  \leq 0.
	\end{align*}
	But the preceding inequality, using the relation $y = W^1 z + b^1$, is equivalent to
	$$
	\begin{bmatrix}
	y \\ 1
	\end{bmatrix}^\top \begin{bmatrix}
	A^2 & A b \\ b^\top A &  b^\top b - 1
	\end{bmatrix} \begin{bmatrix}
	y \\ 1
	\end{bmatrix} \leq 0,
	$$
	which is equivalent to
%	which can be simplified to
%	\begin{align*}
%	||Ay+b||_2 \leq 1.
%	\end{align*}
%	The above inequality is a representation of an ellipsoid. We can equivalently represent this ellipsoid as
%	\begin{align*}
	$(y+A^{-1}b)^\top A^2 (y+A^{-1}b) \leq 1$.
%	\end{align*}
	Using our notation for ellipsoids, this means for all $x \in \mathcal{E}(\mu_x,\Sigma_x)$, we must have $y \in \mathcal{E}(-A^{-1} b, A^{-2})$. 
\end{proof}

\medskip

In Theorem \ref{thm: main result one layer}, we proposed a matrix inequality, in variables $(Q,A,b)$, as a sufficient condition for enclosing the output of the neural network with the ellipsoid $\mathcal{E}(-A^{-1} b, A^{-2})$. We can now use this result to find the minimum-volume ellipsoid with this property. Note that the matrix inequality \eqref{thm: main result one layer 2} is not linear in $(A,b)$. Nevertheless, we can convexify it by using Schur Complements. %We formalize this next.
\begin{lemma} \label{lemma: schure complement}
	The matrix inequality in \eqref{thm: main result one layer 2} is equivalent to the linear matrix inequality (LMI)
	\begin{align} \label{eq: LMI}
	M \!=\! \left[
	\begin{array}{c|c}
	M_1 \!+M_2 \!-\!ee^\top & \begin{matrix}
	0_{n_x \times n_y} \\ {W^1}^\top A \\ {b^1}^\top A \!+\! b^\top
	\end{matrix} \\
	\hline
	\begin{matrix}
	0_{n_y\times n_x} & AW^1 & Ab^1\!+\!b
	\end{matrix} & -I_{n_y}
	\end{array}\right] \preceq 0,
	\end{align}
	in $(\tau,A,Q,b)$, where $e = (0,\cdots,0,1) \in \mathbb{R}^{n_x+n_1+1}$.
\end{lemma}
\begin{proof}
	%Note that $S(A,b)$  in \eqref{thm: main result one layer 2} can be written as
	%\begin{align*}
	%S(A,b) = \begin{bmatrix}
	%A \\ b^\top
	%\end{bmatrix} \begin{bmatrix}
	%A  & b
	%\end{bmatrix}-\begin{bmatrix}
	%0 & 0 \\ 0 &1 
	%\end{bmatrix}.
	%\end{align*}
	%Substituting this back into $M_3$ in \eqref{thm: main result one layer 2}, we can write
	It is not hard to verify that $M_3$ can be written as
	%\begin{align*}
	$M_3 =FF^\top-ee^\top$, 
	%\end{align*}
	where $F$, affine in $(A,b)$, is given by
	$$
	F(A,b) = \begin{bmatrix}
	0_{n_x \times n_y} \\ {W^1}^\top A \\ {b^1}^\top A + b^\top
	\end{bmatrix}.
	$$
	Using this definition, the matrix inequality in \eqref{thm: main result one layer 2} reads
	%
	%\begin{align*}
	$(M_1 + M_2-ee^\top)+ F F^\top \preceq 0$,
	%\end{align*}
	which implies that the term in the parentheses must be non-negative, i.e., 
	$
	M_1 + M_2 -ee^\top \preceq 0.
	$
	Using Schur Complements, the last two inequalities are equivalent to \eqref{eq: LMI}.
\end{proof}
%
%\subsection{Minimum-volume Covering Ellipsoid}

Having established Lemma \ref{lemma: schure complement}, we can now find the minimum-volume covering ellipsoid by solving the following semidefinite program (SDP),
	\begin{alignat}{2} \label{eq: SDP}
	&\text{minimize} \quad && -\log\det(A) \ \text{ subject to } \ \eqref{eq: LMI}.
	\end{alignat}
	where the decision variables are $(\tau,A,Q,b) \in \mathbb{R}_{+} \times \mathbb{S}^{n_y} \times \mathcal{Q} \times \mathbb{R}^{n_y}$.
Since $\mathcal{Q}$ is a convex cone, \eqref{eq: SDP} is a convex program and can be solved via interior-point method solvers. 

\subsection{Multi-layer Case} \label{sec: Multi-layer Case}
For multi-layer neural networks, we can apply the result of Theorem \ref{thm: main result one layer} in a layer-by-layer fashion provided that the input confidence ellipsoid of each layer is non-degenerate. This assumption holds when for all $ 0\leq k \leq \ell-1$ we have $n_{k+1} \leq n_k$ (reduction in the width of layers), and the weight matrices $W^k \in \mathbb{R}^{n_{k+1}\times n_k}$ are full rank. To see this, we note that ellipsoids are invariant under affine subspace transformations such that
$$W^k \mathcal{E}(\mu^k, \Sigma^k) + b^k = \mathcal{E}(W^k \mu^k + b^k, W^k \Sigma^k {W^k}^\top).$$
This implies that $\Sigma_{k+1}:=W^k \Sigma^k {W^k}^\top$ is positive definite whenever $\Sigma^k$ is positive definite, implying that the ellipsoid $\mathcal{E}(\mu_{k+1},\Sigma_{k+1})$ is non-degenerate.
If the assumption $n_{k+1} \leq n_k$ is violated, we can use the compact representation of multi-layer neural networks elaborated in \cite{fazlyab2019safety} to arrive at the multi-layer couterpart of the matrix inequality in \eqref{thm: main result one layer 2}.

\section{Numerical Experiments} \label{sec: numerical experiments}
In this section, we consider a numerical experiment, in which we estimate the confidence ellipsoid of a one-layer neural network with $n_x=2$ inputs, $n_1 \in \{10,30,50\}$ hidden neurons and $n_y = 2$ outputs. We assume the input is Gaussian with $\mu_x = (1,1)$ and $\Sigma_x=\mathrm{diag}(1,2)$. The weights and biases of the network are chosen randomly. We use MATLAB, CVX \cite{grant2008cvx}, and Mosek \cite{mosek} to solve the corresponding SDP. In Figure \ref{fig: confidence ellipsoid}, we plot the estimated $0.95$-level output confidence ellipsoid along with $10^4$ sample outputs. We also plot the image of $0.95$-level input confidence ellipsoid under $f$ along with the estimated $0.95$-level output confidence ellipsoid.
\begin{figure}
	\centering
	\includegraphics[width=0.5\textwidth]{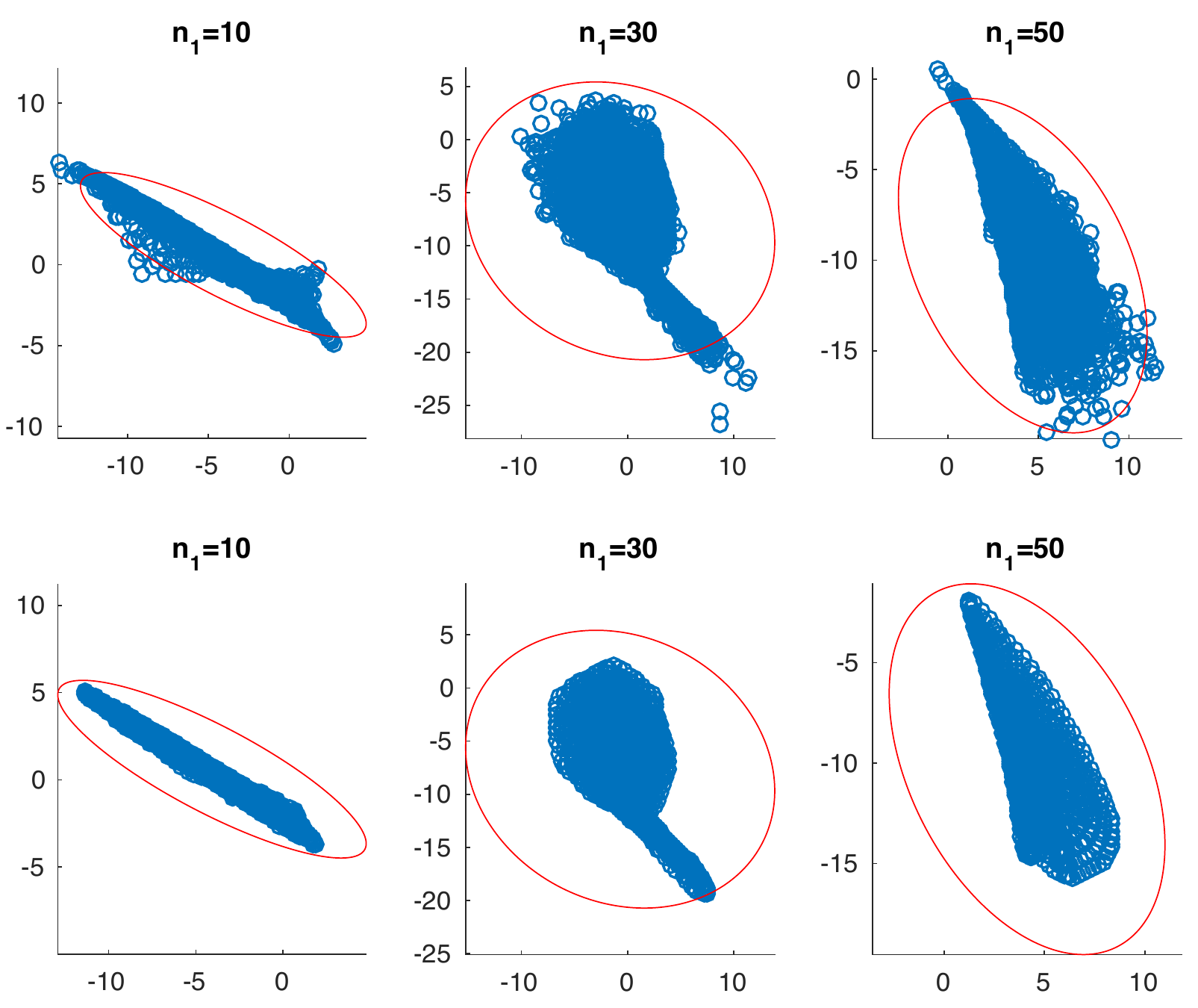}
	\caption{\small Top: the estimated $95\%$ confidence ellipsoid along with $10^4$ samples of the output. Bottom: The image of the $95\%$ input confidence ellipsoid ($f(\mathcal{E}_p)$ with $p=0.95$) and its outer approximation (the output confidence ellipsoid).}
	\label{fig: confidence ellipsoid}
	%\includegraphics[width=\linewidth]{figs/fig4}
	%\caption{\small The image of the $95\%$ input confidence ellipsoid ($f(\mathcal{E}_p)$ with $p=0.95$) and its over-approximation (the output confidence ellipsoid).}
	%\label{fig: confidence ellipsoid 1}
\end{figure}

\section{Conclusions} \label{sec: conclusions}
We studied probabilistic safety verification of neural networks when their inputs are subject to random noise with known first two moments. Instead of analyzing the network directly, we proposed to study the safety of an abstracted network instead, in which the nonlinear activation functions are relaxed by the quadratic constraints their input-output pairs satisfy. We then showed that we can analyze the safety properties of the abstracted network using semidefinite programming. 
%We only discussed two distinct yet related problems: (i) probabilistic safety verification of neural networks; and (ii) confidence ellipsoid estimation. 
%There are other closely related problems that can be tackled by the framework proposed in this paper. For instance, 
It would be interesting to consider other related problems such as closed-loop statistical safety verification and reachability analysis.

\bibliographystyle{ieeetr}
\bibliography{Refs_CDC_2019}

\end{document}